\documentclass{article}
\usepackage{ijcai20}

\usepackage{times}
\usepackage{soul}
\usepackage{url}
\usepackage[hidelinks]{hyperref}
\usepackage[utf8]{inputenc}
\usepackage[small]{caption}
\usepackage{graphicx}
\usepackage{amsmath}
\usepackage{amsthm}
\usepackage{booktabs}
\usepackage{algorithm}
\usepackage{algorithmic}
\newtheorem{theorem}{Theorem}

\title{On the complexity of Winner Verification and Candidate Winner for Multiwinner Voting Rules}

\author{
Chinmay Sonar$^1$ \and
Palash Dey$^2$
\and
Neeldhara Misra$^3$
\affiliations
$^1$University of California Santa Barbara
$^2$Indian Institute of Technology Kharagpur\\
$^3$Indian Institute of Technology Gandhinagar\\
\emails
chinmaysonar96@gmail.com
palash.ju.dey@gmail.com,
neeldhara.m@iitgn.ac.in,
}

\usepackage{amssymb,amsmath}

\newcommand{\pos}{{{{\mathrm{pos}}}}}

\usepackage{amsthm} 
\usepackage[usenames,dvipsnames]{color}

\PassOptionsToPackage{hyphens}{url}
\usepackage[usenames,svgnames,table]{xcolor}
\usepackage{hyperref}
\hypersetup{
     linktocpage,
     colorlinks   = true,
     linkcolor    = Black, 
     urlcolor     = Black, 
     citecolor = Black
}

\usepackage[capitalise,noabbrev]{cleveref}

\usepackage{makecell}
\usepackage{booktabs}
\usepackage{multirow}
\usepackage{mathtools}
\usepackage{thm-restate}
\usepackage{relsize}
\usepackage{enumitem}
\usepackage{comment}
\usepackage{xcolor}
\usepackage{url}
\usepackage{tikz}
\usetikzlibrary{matrix}
\tikzset{node style ge/.style={circle}}
\tikzset{BarreStyle/.style =   {opacity=0.4,line width=5mm,line cap=round,color=#1}}
\tikzset{chosen/.style =   {opacity=0.5,circle,fill=#1!50}}

\newtheorem{lemma}{\bf Lemma}

\newtheorem{definition}{\bf Definition}
\newenvironment{proof}{{\bf Proof:}}{\hfill\rule{2mm}{2mm}}

\usepackage{cuted}
\usepackage{xspace}

\newcommand{\el}{\ensuremath{\ell}\xspace}
\newcommand{\suc}{\ensuremath{\succ}\xspace}

\DeclarePairedDelimiter{\ceil}{\lceil}{\rceil}
\DeclarePairedDelimiter{\floor}{\lfloor}{\rfloor}

\renewcommand{\leq}{\leqslant}
\renewcommand{\geq}{\geqslant}

\renewcommand{\le}{\leqslant}

\newcommand{\YES}{\textsc{Yes}\xspace}
\newcommand{\NO}{\textsc{No}\xspace}

\newcommand{\NP}{NP\xspace}
\newcommand{\NPH}{NP-hard\xspace}
\newcommand{\NPC}{NP-complete\xspace}

\renewcommand{\AA}{\ensuremath{\mathcal A}\xspace}

\newcommand{\CC}{\ensuremath{\mathcal C}\xspace}
\newcommand{\DD}{\ensuremath{\mathcal D}\xspace}

\newcommand{\FF}{\ensuremath{\mathcal F}\xspace}

\newcommand{\LL}{\ensuremath{\mathcal L}\xspace}

\newcommand{\NN}{\ensuremath{\mathcal N}\xspace}

\newcommand{\PP}{\ensuremath{\mathcal P}\xspace}

\newcommand{\ZZ}{\ensuremath{\mathcal Z}\xspace}

\newcommand{\Polytime}{{\textnormal{\textup{P}}}}
\newcommand{\coNP}{coNP\xspace}
\newcommand{\coNPH}{coNP-hard\xspace}

\newcommand{\CCWV}{\textsc{CCWV}}

\newcommand{\MCWlong}{\textsc{Monroe Candidate Winner}}

\newcommand{\MCWi}{\ensuremath{\langle C,V,c,k \rangle}}

\begin{document}

\maketitle

\begin{abstract}
    The Chamberlin-Courant and Monroe rules are fundamental and well-studied rules in the literature of multi-winner elections. The problem of determining if there exists a committee of size $k$ that has a Chamberlin-Courant (respectively, Monroe) score of at most $r$ is known to be NP-complete. We consider the following natural problems in this setting: a) given a committee $S$ of size $k$ as input, is it an optimal $k$-sized committee, and b) given a candidate $c$ and a committee size $k$, does there exist an optimal $k$-sized committee that contains $c$? 

    In this work, we resolve the complexity of both problems for the Chamberlin-Courant and Monroe voting rules in the settings of rankings as well as approval ballots. We show that verifying if a given committee is optimal is coNP-complete whilst the latter problem is complete for $\Theta_{2}^{P}$. We also demonstrate efficient algorithms for the second problem when the input consists of single-peaked rankings. Our contribution fills an essential gap in the literature for these important multi-winner rules.
\end{abstract}

\section{Introduction}

We study preference aggregation in the multiwinner setting. Here, we have a set of voters who express preferences over a collection of candidates, and we are faced with the task of shortlisting a small number of candidates in a manner that is as satisfactory as possible for all the agents involved. This abstraction captures several application scenarios such as choosing a governing body of any institution, deciding which advertisements to show on TV during some program, recommending movies~\cite{LB11}, selecting group of products for promotion \cite{SFL16}, shortlisting candidates for a limited fellowship, etc. 

A fundamental property that one often wants a multiwinner rule to satisfy is {\em proportional representation}. Intuitively speaking, proportional representation means that the fraction of seats that a party receives in the winning committee should be proportional to the fraction of votes it receives. Indeed, two of the most popular multiwinner rules, the Chamberlin-Courant rule \cite{CC83} and the Monroe rule \cite{M95}, are designed to achieve proportional representation as best as one can hope for. In the Chamberlin-Courant (abbreviated CC) rule, we seek winners that can be hypothetically ``assigned to'' voters in such a way that every voter is reasonably satisfied with their alternative. In the Monroe rule, we have the additional requirement that each winning candidate takes on the responsibility of representing roughly the same number of voters. Depending on how one formalizes misrepresentation and how the preferences of the voters are modeled (popularly one of approval ballots or rankings), natural variants of these rules are used in practice.

Bartholdi et al.~\shortcite{bartholdi1989voting} showed that determining winners for many, otherwise excellent, voting rules are \NPH. Prominent examples of such single winner ($k$=1) rules include Kemeny's voting rule~\cite{kemeny1959mathematics}, Lewis Caroll's rule (Dodgson Rule, 1876). Moreover, some of these single winner rules seem to be substantially harder than any \NPC problem --- they are complete for the complexity class $P_{\parallel}^{NP}$~\cite{HSV05}. Papadimitriou and Zachos~\shortcite{PZ82} were the first to introduce the class $P_{\parallel}^{NP}$. Any language in this class can be decided in polynomial time using a polynomial number of \emph{parallel} access to an \NP oracle. Notice that, \emph{parallel} access forbids adaptive queries and only allows `batch' queries to an \NP oracle. 

We recall briefly that the preferences of voters in an election instance are typically solicited as either \emph{rankings} (total orders over candidates) or \emph{approval ballots} (subsets of ``approved'' candidates). The problem of finding a committee whose misrepresentation is bounded by a given threshold is known to be NP-complete for Chamberlin-Courant and Monroe \cite{LB11,PRZ08} in the setting of rankings as well as approval ballots. Moreover, it is both \NPH and \coNPH even to decide whether some given candidate belongs to an optimal CC committee~\cite[Corollary 3]{BFKNST17} in the setting of rankings. Our main contribution in this work is to completely settle the complexity of two natural versions of the winner determination question in the context of the two fundamental multiwinner rules, Chamberlin-Courant and Monroe. We address these problems in the settings of both rankings and approval ballots.


\textbf{Our Contribution.} 
We consider the following problems for both the Chamberlin-Courant and Monroe rules, in the setting of approval ballots and rankings. In the \textsc{Winner Verification} problem, we want to know if a proposed committee is optimal, and in the \textsc{Candidate Winner} problem, we are given a candidate $c$ and a committee size $k$, and the question is if there exists an optimal $k$-sized committee containing $c$. 

\emph{Winner Verification.} Our first set of contributions is for the \textsc{Winner Verification} problem; we show that it is complete for the complexity class \coNP{}. In this case, the membership is easy to establish. For a given committee, observe that it is easy to compute it's score with respect to the Chamberlin-Courant rule (and also the Monroe rule, although this is less straightforward). Thus, our coNP certificate is simply a ``rival" committee with a better score. We remark, as an aside, that this is in contrast with rules such as Dodgson for which computing Dodgson score of a given candidate is intractable.  To show hardness for coNP, we reduce from the complement of \textsc{Hitting Set} problem in different ways depending on the setting. For showing the hardness of \emph{Monroe Rule} we employ a variant where the elements enjoy uniform occurrences among the sets. Apart from settling the complexity of fundamental question of winner verification, our contribution identifies a natural coNP-complete problem, in particular, one that is \emph{not} merely the complement of a natural NP-complete problem. 

\emph{Candidate Winner.} This problem was recently shown to be both NP-hard and coNP-hard~\cite{BFKNST17} for the Chamberlin-Courant voting rule in the setting of rankings. We demonstrate here that the problem is complete for $P_{\parallel}^{NP}$ (or, equivalently, $\Theta_{2}^{P}$) for both the Chamberlin-Courant and Monroe rules, in the setting of rankings as well as approval ballots. All of these results involve reductions from the \textsc{Vertex Cover Member} problem. Although these reductions are executed in a similar spirit, the different settings do require non-trivial techniques in the constructions. 

Our main contributions are summarized below. We refer the reader to the next section for the relevant terminology.


\begin{theorem}
    \label{thm:wv-cc}
    \textsc{Winner Verification} for Chamberlin-Courant and Monroe is coNP-complete in the setting of approval ballots and rankings. In the latter setting, the result holds for the $\ell_1$ and $\ell_\infty$ Borda misrepresentation functions.
\end{theorem}

\begin{theorem}
    \label{thm:wv-cc}
    \textsc{Candidate Winner} for Chamberlin-Courant and Monroe is complete for $\Theta_{2}^{P}$ in the setting of approval ballots and rankings. In the latter setting, the result holds for the $\ell_1$ and $\ell_\infty$ Borda misrepresentation functions.
\end{theorem}

Each of the statements above addresses six distinct scenarios. Due to lack of space, we give complete proofs for three of these settings, which we believe to be representative of the overall flavor of the arguments. For ease of presentation, all of the results presented here focus on the Chamberlin-Courant rule, and we will briefly explain the techniques involved in obtaining the analogous results for the Monroe rule. We also show that the \textsc{Candidate Winner} problem for Chamberlin-Courant can be solved in polynomial time on single-peaked voting profiles.

\newpage 
\section{Preliminaries}

For a positive integer $\el$, we denote the set $\{1, \ldots, \el\}$ by $[\el]$. We first define some general notions relating to voting rules. Let $V$ be a set of $n$ {\em voters} and $C$ be a set of $m$ {\em candidates}. If not mentioned otherwise, we denote the set of candidates, the set of voters, the number of candidates, and the number of voters by $C$, $V$, $m$, and $n$ respectively. Every voter $v$ has a {\em preference} $\suc_v$ which is typically a complete order over the set $C$ of candidates (rankings) or a subset of approved candidates (approval ballots). An instance of an election consists of the set of candidates $C$ and the preferences of the voters $V$, usually denoted as $E = (C,V)$.

We now recall some definitions in the context of rankings. We say voter $v$ prefers a candidate $x\in C$ over another candidate $y\in C$ if $x\suc_v y$. For a ranking $\succ$, pos$_\succ(c)$ is given by one plus the number of candidates ranked above $c$ in $\succ$. In particular, if there are $m$ candidates and $c$ is the top-ranked (respectively, bottom-ranked) candidate in the ranking $\succ$, then pos$_\succ(c)$ is one (respectively, $m$). We denote the set of all preferences over $C$ by $\LL(C)$. The $n$-tuple $(\suc_v)_{v\in V} \in\LL(C)^n$ of the preferences of all the voters is called a {\em profile}. Without loss of much rigor, we note that a profile, in general, is a multiset of linear orders. For a subset $M\subseteq V$, we call $(\suc_v)_{v\in M}$ a sub-profile of $(\suc_v)_{v\in V}$. For a subset of candidates $D \subseteq C$, we use $\PP\vert_D$ to denote the projection of the profile on the candidates in $D$ alone. The definitions of profiles, sub-profiles, and projections are analogous for approval ballots.




\paragraph*{Chamberlin-Courant for Rankings.} The Chamberlin--Courant voting rule is based on the notion of a {\em dissatisfaction} or a {\em misrepresentation} function. This function specifies, for each $i\in[m]$, a voter's dissatisfaction $\alpha^m(i)$ from being represented by the candidate she ranks in position~$i$. A popular dissatisfaction function is Borda, given by $\alpha^m(i) = i-1$.

We now turn to the notion of an assignment function. Let $k \leq m$ be a positive integer. A {\em $k$-CC-assignment function} for an election $E = (C,V)$ is a mapping $\Phi \colon V \rightarrow C$ such that $| \Phi(V) | = k$, where $\Phi(V)$ denotes the image of $\Phi$.  For a given assignment function $\Phi$, we say that voter $v\in V$ is \emph{represented} by candidate $\Phi(v)$ in the chosen committee. There are several ways to measure the quality of an
assignment function $\Phi$ with respect to a dissatisfaction function $\alpha:[m]\longrightarrow \mathbb{R}$; and we will use the following:
\begin{enumerate}
\item $\ell_1(\Phi,\alpha) = \sum_{v\in V} \alpha(\pos_{v}(\Phi(v)))$, and
\item $\ell_\infty(\Phi,\alpha) = \max_{v\in V}\alpha(\pos_{v}(\Phi(v)))$.
\end{enumerate}
Unless specified otherwise, $\alpha$ will be the Borda dissatisfaction function described above. We are now ready to define the Chamberlin-Courant voting rule.

\begin{definition}[Chamberlin-Courant] For $\ell\in\{\ell_1, \ell_\infty\}$, the {\em $\ell\hbox{-}$CC voting rule} is a mapping that takes an election $E = (C,V)$ and a positive integer $k$ with $k\le |C|$ as its input, and returns the images of all the $k$-CC-assignment functions $\Phi$ for $E$ that minimizes  $\ell(\Phi,\alpha)$. 
\end{definition}

\paragraph*{Chamberlin-Courant for Approval Ballots.} Recall that an approval vote of a voter $v$ on the set of candidates $C$ is some subset $S_{v}$ of $C$ such that $v$ approves all the candidates in $S_{v}$. We define the misrepresentation score of a $k$-sized committee $W$ as the number of voters which do not have any of their approved candidates in $W$ (i.e. $W \cap S_{v} = \phi$). Hence the optimal committees under approval Chamberlin-Courant are the committees which maximize the number of voters with at least one approved candidate in the winning committee~\cite{LS18}. 

We are now ready to describe the questions that we study in this paper. The first problem is \textsc{Chamberlin-Courant Winner Verification} (CCWV). Here, the input is an election $E = (C,V)$ and a subset $S$ of $k$ candidates. The question is if $S$ is a winning $k$-sized CC-committee for the election $E$, in other words, does $S$ achieve the best Chamberlin-Courant score in the given election among all committees of size $k$? 

In the second problem, given an election $E=(C,V)$, a committee size $k$, and a candidate $c \in C$, we ask if $c$ belongs to \emph{some} optimal $k$-sized committee, in other words, if there exists $S \subseteq C$ such that $c \in S, |S|=k$, and $S$ is a winning CC committee. We refer to this as the \textsc{Chamberlin-Courant Candidate Winner} problem (CCCW).


%

We now turn to the definition of the Monroe voting rule~\cite{M95}. Let ${\Phi^{-1}}(c)$ for $c \in C$ denote the set of voters \emph{represented} by $c$.

\begin{definition}[Monroe]  For $\ell\in\{\ell_1, \ell_\infty\}$, the {\em $\ell\hbox{-}$Monroe voting rule} is a mapping that takes an election $E = (C,V)$ and a positive integer $k$ with $k\le |C|$ as its input, and returns the image of any of the $k$-Monroe-assignment functions $\Phi$ such that $|\Phi^{-1}(c)|$ is either $\floor{\frac{n}{k}}$ or $\ceil{\frac{n}{k}}$ where $c \in C$ for $E$ that minimizes  $\ell(\Phi,\alpha)$.
\end{definition}

We note that \textsc{Monroe Winner Verification} (MWV) and \textsc{Monroe Candidate Winner} (MCW) are defined in the natural way. We also recall the definitions of \textsc{3-Hitting Set} and its \emph{complement}. In the \textsc{3-Hitting Set} problem, we are given a ground set $\mathcal{U}$, a family $\mathcal{F}$ of three-sized subsets of $\mathcal{U}$, and an integer $k$, and the question is if there exists $S \subseteq \mathcal{U}$ of size at most $k$ that intersects every set in $\mathcal{F}$, i.e: $\forall$ $F \in \mathcal{F}$, $S \cap F \neq \phi$. In the \textsc{c-3-Hitting Set} problem, the input is the same, and is a \textsc{Yes}-instance if and only if there is no hitting set of size $k$; in other words, if for each $S \subseteq \mathcal{U}$ with $|S| \leq k$, there exists some $F_S \in \mathcal{F}$ such that $S \cap F_S = \phi$. We recall that \textsc{3-Hitting Set} is a classic NP-complete problem, and c-\textsc{3-Hitting Set} is \emph{co-NP} complete.



\textbf{The Class $P^{NP}_{\parallel}$ ($\Theta_{2}^{P}$)} \label{subsec:intro_to_hard_class}
The class $P_{\parallel}^{NP}$ is the class of problems solvable using \Polytime{} machine having parallel access to an \NP{} oracle. The class $\Theta_{2}^{P}$ was introduced in \cite{PZ82} and named in \cite{wagner90}. The class $\Theta_{2}^{P}$ was shown to be equivalent to $P^{NP}_{\parallel}$ by Hemachandra \shortcite{H89}. The \textsc{Vertex Cover Member} problem is the following. Given a graph $G:=(V,E)$ and a vertex $w \in V$, the question is if there exists a minimum sized vertex cover containing $w$. The problem was shown to be complete for $P_{\parallel}^{NP}$ by \cite{HSV05}.

\newpage

\section{Winner Verification Problems}

In this section, we show the coNP-completeness of \textsc{Chamberlin-Courant Winner Verification} in the setting of rankings for the $\ell_1$-Borda misrepresentation score. The argument for membership is, in brief, the following: a rival committee with a better misrepresentation score is a valid certificate for the \NO{} instances of \CCWV{}. This is an efficiently computable certificate since it is easy to compute the Chamberlin-Courant score of a given committee. We now turn to the proof of hardness. 


\begin{theorem}
    \label{thm:CC-Winner-rankings-coNP}
    \textsc{Chamberlin-Courant Winner Verification} is coNP-hard in the setting of rankings for the $\ell_1$-Borda misrepresentation score.
\end{theorem}

\begin{proof} We show a reduction from c-\textsc{3-Hitting Set} to the \textsc{CC-Winner} problem. Let $\langle  U,\FF; k \rangle$ be an instance of c-\textsc{3-Hitting Set} with $n$ elements in the universe $U$ and $m$ sets of size three in the family $\FF$. We construct a profile $\PP$ over alternatives $\AA$ as follows. First, we introduce one candidate corresponding to each element of the universe $U$, $k$ ``dummy'' candidates, and a large number of ``filler'' candidates, that is:
    \[ \AA := \underbrace{\{c_u ~|~ u \in U \}}_{\CC} ~\cup~ \underbrace{\{d_1, \ldots, d_k\}}_{\DD}  ~\cup~ \underbrace{\{z_1, \ldots, z_t\}}_{\ZZ}, \]
    where $t = 3(mk)^2$. Also, for every $1 \leq i \leq k$, and for every $X \in \FF$, introduce a vote $v(i,X)$ that places the candidates corresponding to the elements in $X$ in the top three positions, followed by $d_i$, followed by $3mk$ candidates from $Z$. We ensure that we use distinct candidates from $Z$ in the top $3mk + 4$ positions of all the voters, in other words, no candidate from $Z$ appears twice in the top $3mk + 4$ positions. Note that $t$ is chosen to be large enough to make this possible. This is followed by the candidates in $U\setminus X$ ranked in an arbitrary order followed by the remaining filler candidates, also ranked in an arbitrary order.
    
    In this instance, note that a committee corresponding to a hitting set has a score of at most $2mk$, while the score of the committee $\DD$ is $3mk$. In the constructed instance, we now ask if the committee $\mathcal{D}$ consisting of $k$ dummy candidates is a winning committee. This completes the construction of the instance. We now turn to the equivalence of two instances.
    
    In the forward direction, suppose we have a \YES{} instance of c-\textsc{3-Hitting Set}. This implies that there does not exist any hitting set of size at most $k$. Recall that misrepresentation score for committee consisting of a hitting set is at most $2mk$, while noting that any such committee must have size greater than $k$. Now, we show that for all other committees of size at most k, the misrepresentation score is greater than $3mk$.
    \begin{lemma}
    \label{lem:mis-rep-CC-rankings-l1}
    Consider an instance $\langle \AA,V,\DD \rangle$ of CC-winner Verification based on a \textsc{Yes}-instance of c-3-Hitting Set $\langle  U,\FF,k \rangle$. For any feasible committee $C' \subseteq \AA$ of size $k$ different from $\DD$, the $\ell_1$-Borda misrepresentation score of $C'$ is greater than $3mk$.
    \end{lemma}
    \begin{proof} \emph{(of Lemma~\ref{lem:mis-rep-CC-rankings-l1}.)}
    Let $U'$, $D'$ and $Z'$ denote, respectively, the candidate subsets $C' \cap \CC$, $C' \cap \DD$ and $C' \cap \ZZ$. Since $C'$ is different from $\DD$, there is at least one candidate from $\DD$ that does not belong to $C'$ (the only other possibility is that $C'$ is a superset of $\DD$, but this is not possible since $|C'| = |\DD| = k$). Without loss of generality, suppose $d_1 \notin C'$. Now consider the votes given by $V' := \{v(1,X) \mid X \in \FF\}.$ We claim that there are at least $|Z'| + 1$ voters in $V'$ whose misrepresentation score for the committee $C'$ is strictly greater than three. Indeed, if not, then it is straightforward to verify that $U'$ combined with an arbitrarily chosen element from each set not hit by $U'$ comprises a subset of size at most $|U'| + |Z'| \leq k$ which intersects every set in $\FF$, contradicting our assumption that $\FF$ has no hitting set of size at most $k$. To see this, observe that every vote in $V'$ that has a misrepresentation score of three or less is necessarily represented by a candidate from $U'$, since $d_1 \notin C'$, and therefore, the sets corresponding to all of these votes are hit by $U'$, and the remaining sets can be hit ``trivially'' since there are  at most $|Z|$ of them. Now consider the voters who have a ``high'' misrepresentation score:
    $V'' := \{v(1,X) \mid X \in \FF \mbox{ and } \tau(v(1,X),C') > 3\}.$
    By the argument in the previous paragraph, we have that $|V''| > |Z|$. Recalling that every vote has distinct filler candidates in the top $3mk$ positions after $d_i$, by the pigeon-hole principle, we conclude that there is at least one vote $v(1,X)$ in $V''$ such that $Z_X \cap Z = \emptyset$, where $Z_X$ denotes the filler candidates that appear in the top $3mk + 4$ positions of the vote $v(1,X)$. Since the candidates occupying the top four positions of this vote do not belong to $C'$ either, it follows that the misrepresentation score of $v(1,X)$ for $C'$ is greater than $3mk$, and this concludes our argument.
    \end{proof}
    

    The committee $\mathcal{D}$ has a misrepresentation score of $3mk$. Using Lemma \ref{lem:mis-rep-CC-rankings-l1}, since $\FF$ has no hitting set of size at most $k$, we have that $\mathcal{D}$ is a winning committee among all feasible committees, as desired.
    
    In the reverse direction, we start with the assumption that $\mathcal{D}$ is a winning committee. Therefore, the optimal misrepresentation for the constructed election instance is $3mk$. Observe that if there exists a hitting set $S$ of size at most $k$, then the committee $C'$ formed using the corresponding candidates of hitting set will have misrepresentation score of at most $2mk$, as discussed above. Thus, $\DD$ would not be a committee, a contradiction --- and this implies that $\langle  U,\FF; k \rangle$ was indeed a \YES{}-instance of c-\textsc{3-Hitting Set}. This completes the argument of equivalence.
\end{proof}

For $\ell_{\infty}$-CC, we again reduce from the complement of the \textsc{3-Hitting Set} problem with a similar construction. We introduce votes corresponding to sets in the family, where the top three candidates are the candidates corresponding to the elements contained in the set, and the fourth candidate is a dummy candidate. Once we construct $k$ blocks with distinct dummy candidates in the fourth position, the possible misrepresentations will play out in an analogous fashion. For obtaining these three results in the setting of Monroe, while the reduction is similar, we have to reduce from a variant of the Hitting Set problem with additional structure. 

%

\section{Candidate Winner Problems}

In this section, we turn to the \textsc{Monroe Candidate Winner} problem. Recall that the input is \MCWi{}, and the question is if there exists an optimal Monroe committee of size $k$ containing $c$. We demonstrate that the problem is complete for $\Theta_{2}^{P}$ in the setting of rankings for both the $\ell_1$ Borda misrepresentation function. The argument for the case of approval ballots for this problem are in a similar spirit, and are deferred to a full version of this paper. 

We first consider the case of the $\ell_1$ Borda misrepresentation. Our focus here will be on showing hardness, and we informally justify the claim for membership. We use oracle queries to a variant of MCW where we additionally demand for the committee to achieve a particular target misrepresentation score. Note that the worst possible misrepresentation score in an instance with $m$ candidates and $n$ voters is $mn$. Thus, by guessing this target score, we can find the score of the optimal Monroe committee that contains $c$ and the score of the optimal Monroe committee, and comparing these answers the question of whether there exists an optimal Monroe committee of size $k$ containing $c$. We now turn to the reduction to demonstrate hardness. 

\begin{theorem}
\label{thm:Monroe-Candidate-Winner-hardness-l1}
\MCWlong{} is $\Theta_{2}^{P}$-hard for the $\ell_{1}$ Borda misrepresentation function.
\end{theorem}

\begin{proof}
We reduce from the $\Theta_{2}^{P}$-complete problem \textsc{Vertex Cover Member}. Recall that we are given a graph $G:=(V,E)$ (with $n$ vertices and $m$ edges), and a vertex $w \in V$, the question is if there exists a minimum sized vertex cover containing $w$. Given an instance $\langle G:=(V,E), w \rangle$ of \textsc{Vertex Cover Member} we construct an instance of \textsc{CC Candidate Winner} as follows. Let the set of candidates be $C := C_{v} \cup D \cup D' \cup S$, where $C_{v}$ denotes the set of $n$ candidates corresponding to vertices of $G$ and $D$ and $D'$ denote \emph{type I} and \emph{type II} dummy candidates respectively. Let $\Delta$ denote a set of $n^{4}m$ type I dummy candidates, and $\Delta'$ denote a set of $2(n^{4}m)$ type II dummy candidates. We note that the subsets of dummy candidates specified explicitly in different votes are always chosen so that there are no repeated dummy candidates in the explicitly defined blocks, in other words, the chosen dummy candidates are always distinct. Also, $D$ ($D')$ is the union of all the $\Delta$'s ($\Delta'$'s) respectively specified in the profile, which is given by the following five blocks of voters:
\begin{itemize}
    \item \textbf{Block 1}: We construct $m$ votes corresponding the edges in $G$. For an edge $(u,v)$ we add:
$$c_{u} \succ c_{v} \succ \Delta' \succ C_{v}\setminus\{c_{u},c_{v}\} \succ \mbox{rest}$$
where ``rest'' denotes the set of remaining candidates placed in an arbitrary order. 
    \item \textbf{Block 2:} For the desired vertex $w$ from the \textsc{Vertex Cover Member} instance, we pick an arbitrary edge incident on $w$ (say $(w,x)$) in $G$, and add $m+1$ copies of the following vote:
    \[ c_{w} \succ c_{x} \succ \Delta' \succ C_{v}\setminus\{c_{w},c_{x}\} \succ \mbox{rest}\]
    \item \textbf{Block 3:} We add $n$ votes of the form:
    \[ \Delta  \succ c_{w} \succ C_{v}\setminus c_{w} \succ \mbox{rest}\]
    \item \textbf{Block 4:} For each candidate $d_j \in D_s$, we add the vote:
    \[ d_j \succ \Delta' \succ \mbox{rest}. \]
    \item \textbf{Block 5:} Let $D_{\alpha}$ be a subset of dummy candidates $d \in D$ such that $d$ appears in the top position for one of the votes in Block 3. Note that $|D_{\alpha}|=n$. Further, let $\NN=2(m+n+1)$. For each $v \in \{C_v \cup D''\}$ and $\ell \in [\NN]$, we add the following vote:
    \[  v \succ D_{s} \succ \Delta' \succ \mbox{rest}\]
\end{itemize}

In the constructed \textsc{Monroe Candidate Winner} instance, we ask if there exists an optimal committee of size $2n+1$ containing $c_{w}$. This completes the construction for our reduction. Before showing the equivalence of the two instances, we establish the following lemma.

\begin{lemma}
\label{lem:l1-Monroe-candidate-winner-structure}
Let $q$ be the size of an optimal vertex cover in $G$. Then, following holds for any optimal committee $C'$ of size $2n+1$ in the constructed election instance:
\begin{enumerate}
    \item $C'$ does not contain any $d' \in D'$.
    \item $S \subset C'$.
    \item $C'$ contains exactly $q$ candidates corresponding to an optimal vertex cover.
\end{enumerate}
\end{lemma}

\begin{proof}
First, we analyze the Monroe score of a committee $\CC$ which contains all $n+1$ special candidates $S$, $q$ candidates corresponding to an optimal vertex cover $(S')$, and the remaining $(n-q)$ candidates from $D$ which appear in the top positions of $(n-q)$ votes in \emph{Block 3}. Note that in any Monroe committee $C'$, each candidate represents exactly $\NN$ votes. 
We now describe the Monroe assignment for $\CC$. Each vote in block 1 and 2 is represented by one of the top two candidates such that the corresponding vertex $(v)$ belongs to the vertex cover $(S')$. The misrepresentation for $\CC$ in Block 1 is at most $m$, and in Block 2 it is at most $m+1$. In Block 3, exactly $n-q$ votes are represented by their first choice. For those votes that are not represented by the top candidate already, the misrepresentation for $\CC$ is at most $(|\Delta|+n-1)$ per vote since all the votes in Block 3 are represented by the candidate in $C_{v}$ for that vote in the worst case. In block 4, all votes are represented by their top choice yielding zero misrepresentation. Votes in block 5 are represented as follows: For each candidate $c_{i}$ corresponding to a vertex $u \in S'$,  if $c_{u}$ represents $t$ votes from first 3 blocks, then $c_{u}$ also represents $(\NN-t)$ votes from block 5 among the ones she appears at the first position. Similarly, for $d \in \{D \cap \CC\}$, $d$ represents $\NN-1$ votes among the ones she appears at the top position. Next, each special candidate $s \in S$ represent $(\NN-1)$ votes in block 5, yielding misrepresentation score at most $(n+1)(\NN-1)$ for each $s$. Hence, the total misrepresentation for $\CC$ is strictly less than $n^4m +3mn^{2}+ 3n^{3} < 2(n^{3}m^{2}+n^{4}m)$ for large enough $n$.

Towards showing the first statement, consider a committee $C^{*}$ which contains $d' \in D'$. In any Monroe assignment, $d'$ has to represent $\NN$ votes. Observe that $d'$ appears in first $(n^{4}m)$ positions exactly once, hence, $\mbox{misrepresentation}(C^{*})>\mbox{misrepresentation}(\CC)$. To show the second statement, consider a committee $C^{*}$ which excludes a special candidate $s \in S$. It is easy to see that $\mbox{misrepresentation}(C^{*}) > \mbox{misrepresentation}(\CC)$ even if we only consider misrepresentation from a single vote from block 4 with $s$ at the first position.  

We now turn to statement 3. Now, let $C^{*}$ be an optimal committee which does not contain $q$ candidates corresponding to some optimal vertex cover. We use statements 1 and 2 to analyze following two cases:
\begin{itemize}
    \item \textbf{$|C^{*} \cap C_{v}| > q$:} In this case, $C^{*}$ contains at most $(n-q-1)$ candidates from $D$. Hence, the misrepresentation of $C'$ from Block 3 is at least $|\Delta| \times (q+1)$ which is greater than the misrepresentation score for $\CC$. This contradicts the optimality of $C'$.
    \item \textbf{$|C' \cap C_{v}| \leq q$:} Since the size of an optimal vertex cover is $q$, any committee $C'$ with at most $q$ candidates from $C_{v}$ does not include any candidates corresponding to the endpoint of at least one edge due to the case we are in (i.e. $C'$ does not contain candidates corresponding to optimal vertex cover). Hence, $C'$ incurs a misrepresentation of at least $|\Delta'|$ from one of the votes in Block 1 which implies $\mbox{misrepresentation}(C') >  \mbox{misrepresentation}(\CC)$.
\end{itemize}
This completes the proof for \cref{lem:l1-Monroe-candidate-winner-structure}.
\end{proof}

We now turn to the proof of equivalence. In the forward direction, given an optimal vertex cover of size $q$ containing $w$, we construct an optimal committee $C'$ by choosing $q$ candidates corresponding to the vertex cover, $(n-q)$ candidates from the set $D$ which appears in the top position of exactly $(n-q)$ votes from Block 3, all $n+1$ special candidates $S$. We compute the Monroe assignment of $C'$ is the same way we did for committee $\CC'$ in \cref{lem:l1-Monroe-candidate-winner-structure}. By \cref{lem:l1-Monroe-candidate-winner-structure}, we already know that any optimal committee must contain all candidates from $S$, and candidates corresponding to an optimal sized vertex cover. Therefore, it suffices to show that committees corresponding to optimal vertex covers not containing $c_w$ are not optimal. Indeed, this follows from the fact that in Block 2, $c_{w}$ is the top candidate in exactly $(m+1)$ votes, and in Block 3, $c_{w}$ leads all other candidates from the set $C_{v}$. Hence, it is easy to verify that an optimal committee must contain $c_{w}$. 

In the reverse direction, given an optimal committee $C'$ containing $c_{w}$, we need to construct an optimal vertex cover for $G$ which  includes the vertex $w$. Since $C'$ is optimal, using \cref{lem:l1-Monroe-candidate-winner-structure} we know $C' \cap C_{v}$ is an optimal vertex cover of $G$. Since we are given that $c_w \in C'$, we have that the vertex cover corresponding to $C'$ is an optimal vertex cover containing $c_w$, as desired.  
\end{proof}

Now, we show that CCCW is hard for $\Theta_{2}^{P}$ in the setting of rankings for the $\ell_\infty$ Borda misrepresentation function. We recall that an analogous result for the $\ell_\infty$ Borda misrepresentation function was shown in~\cite{BFKNST17}. The argument for membership is similar to the previous case and is omitted for brevity.

\begin{theorem}
    \label{thm:CC-Candidate-Winner-hardness-linfty}
    CC-Candidate Winner is $\Theta_{2}^{P}$-hard for the $\ell_{\infty}$ Borda misrepresentation function.
    \end{theorem}
    \begin{proof}
    As before, we reduce from \textsc{Vertex Cover Member}. Given an instance $\langle G:=(V,E), w \rangle$ of \textsc{Vertex Cover Member} we construct an instance of \textsc{CC Candidate Winner} as follows. Let the set of candidates be $C := C_{v} \cup D \cup D'$, where $C_{v}$ denotes the set of $n$ candidates corresponding to vertices of $G$, and $D$ and $D'$ denote \emph{type I} and \emph{type II} dummy candidates respectively. Let $\Delta$ denote a set of $m+n+1$ type I dummy candidates, and $\Delta'$ denote a set of $2$ type II dummy candidates. We note that the subsets of dummy candidates specified explicitly in different votes are always chosen so that there are no repeated dummy candidates in the explicitly defined blocks, in other words, the chosen dummy candidates are always distinct. We construct the set of voters as following three blocks:
    \begin{itemize}
        \item \textbf{Block 1}: We construct $(n+2)$ copies of each of $m$ votes corresponding the edges in $G$. For an edge $(u,v)$ we add $(n+2)$ copies of:
    \[ c_{u} \succ c_{v} \succ \Delta \succ C_{v}\setminus\{c_{u},c_{v}\} \succ \mbox{rest},\]
    where \emph{rest} denotes the set of remaining candidates in some arbitrary order. 
        \item \textbf{Block 2:} We add the following $n$ votes:
        \begin{align*}
            & v_{1} := \Delta' \succ d_{1,1} \succ d' \succ \Delta \succ C_{v} \succ  \mbox{rest}\\
            & \vdots\\
            & v_{i} := \Delta' \succ d_{i,1} \succ \ldots \succ d_{i,i} \succ d' \succ \Delta \succ C_{v} \succ \mbox{rest}\\
            & \vdots\\
            & v_{n} := \Delta' \succ d_{n,1} \succ \ldots \succ d_{n,n} \succ d' \succ \Delta \succ C_{v} \succ \mbox{rest}
        \end{align*}
         
        \item \textbf{Block 3:} We also add $(n+2)$ copies of the following vote to force $d'$ in any optimal committee:
        \[ d' \succ \Delta \succ \mbox{rest} \]
        
    \end{itemize}
    
    In the constructed \textsc{CC Candidate Winner} instance, we ask if there exists an optimal committee of size $n+1$ containing $c_{w}$. This completes the construction for our reduction. We now state a lemma analogous to \cref{lem:l1-Monroe-candidate-winner-structure}, whose proof is omitted due to space constraints.
    
    \begin{lemma}[$\star$]
    \label{lem:linfty-CC-candidate-winner-structure}
    Let $q$ be the size of optimal vertex cover in $G$ such that $q \geq 2$. Then, following holds for any committee $C'$ of size $n+1$ in the constructed election instance:
    \begin{enumerate}
        \item Any optimal committee contains candidate $d'$.
        \item If $C'$ contains $d \in D$, then $C'$ is not an \emph{optimal} committee.
        \item If $C'$ does not contain exactly $q$ candidates corresponding to an optimal vertex cover, then $C'$ is not optimal.
    \end{enumerate}
    \end{lemma}

    We now briefly sketch the proof of equivalence. In the forward direction, given an optimal vertex cover $V'$ of size $q$ containing $w$, we construct a committee $\CC$ as described in the first paragraph of the proof of \cref{lem:linfty-CC-candidate-winner-structure}. The fact that the proposed committee $\CC$ is indeed optimal is based on the observations about the structure of optimal committees in \cref{lem:linfty-CC-candidate-winner-structure} and the construction, and is easy to verify.
    In the reverse direction, let $C'$ be an optimal committee containing $c_{w}$. Using \cref{lem:linfty-CC-candidate-winner-structure}, we know that $C'$ contains candidates corresponding to optimal vertex cover. Hence, we can recover an optimal vertex cover containing $c_{w}$, concluding the argument.
\end{proof}
    
    

    We show analogous results for the CC voting rule in the context of approval ballots and the Monroe voting rule in all the remaining settings.
    
    

\section{Concluding Remarks}
We have addressed the problems of \textsc{Winner Verification} and \textsc{Candidate Winner} in the setting of multiwinner voting. We have resolved the complexity of both the problems for the Chamberlin-Courant and Monroe voting rules in various scenarios. In particular, for rankings, we have considered both the $\ell_1$ and the $\ell_\infty$ notions of misrepresentation, and we have also studied variants of these rules in the context of approval ballots. We have showed that verifying if a given committee is optimal is coNP-complete whilst the second problem is complete for $\Theta_{2}^{P}$ in all the twelve cases. Overall, our results comprehensively settle the complexity of these two problems in the general setting. These outcomes primarily serve the purpose of deepening our understanding of where these problems lie in the complexity-theoretic landscape. Further, since the \textsc{Winner Verification} family of problems are complete for $\Theta_2^P$, our results also hint that natural heuristics for the question are unlikely to perform well in practice. Indeed, investigating the performance of heuristics (by possibly adapting greedy approaches for finding optimal committees and forcing the choice of a desired candidate) would be an interesting direction for complementing our theoretic considerations. 

Another natural direction for further thought is the setting of restricted domains, which have received much attention for capturing structure in real-world data sets and for providing natural ``islands of tractability'' for several hard voting problems~\cite{elkind2017structured}. Indeed, although determining optimal committees for the Monroe rule remains intractable even in the setting of single-crossing profiles~\cite{skowron2013complexity}), we can find an optimal Chamberlin-Courant committee efficiently if the input is single-peaked~\cite{BSU13} or single-crossing~\cite{skowron2013complexity}. Further, the rule is tractable also for structured approval ballots~\cite{elkind2015structure}. With this background, it would be interesting to explore the complexity of the problems we study in the setting of restricted domains. The \textsc{Winner Verification} problems are tractable whenever the naturally associated \textsc{Winner Determination} problem is tractable, but the \textsc{Candidate Winner} problem is less immediate to resolve. In the single-peaked setting, with the $\ell_1$ Borda misrepresentation score, the \textsc{Candidate Winner} problem can be resolved by adding several dummy voters who place the desired candidate at the top position, and comparing the optimal CC scores of the original and modified instances. The situation for other restricted domains remains open. 

\newpage

\newpage

\bibliographystyle{named}
\bibliography{ijcai20}

\end{document}